\documentclass[12pt]{article}
\usepackage[english]{babel}
\usepackage{graphicx} 
\usepackage{latexsym}

\newtheorem{theo}{Theorem}[section]
\newtheorem{pr}[theo]{Proposition}
\newtheorem{lemma}[theo]{Lemma}
\newtheorem{defi}[theo]{Definition}

\input amssym.def
\newsymbol\rtimes 226F
\newfont{\nset}{msbm10}
\newcommand{\ns}[1]{\mbox{\nset #1}}

\def\Z{\ns Z}

\def\Z{\ns{Z}}

\def\dg{\delta}

\def\dist{\partial}
\def\ecc{\mathop{\varepsilon}\nolimits}

\def\mod{\mathop{\rm mod}\nolimits}
\def\>{\mathop{\rightarrow}\nolimits}

\def\p{\mbox{\boldmath $p$}}
\def\r{\mbox{\boldmath $r$}}
\def\u{{\mbox {\boldmath $u$}}}

\def\x{\mbox{\boldmath $x$}}
\def\y{\mbox{\boldmath $y$}}
\def\z{\mbox{\boldmath $z$}}
\def\vecalpha{\mbox{\boldmath $\alpha$}}

\def\vec0{\mbox{\bf 0}}

\title{Deterministic Hierarchical Networks
\thanks{Corresponding author: M.A. Fiol, Dept. Matem\`{a}tica Aplicada IV,
Universitat Polit\`{e}cnica de Catalunya, Jordi Girona 1-3,
08034 Barcelona (Catalonia), Tel: +34 93
4015993, Fax: +34 93 4015981, e-mail: {\tt fiol@ma4.upc.edu}} }

\author{L. Barri\`ere, F. Comellas, C. Dalf\'o, M.A. Fiol\\
{\small Departament de Matem\`atica Aplicada IV }\\
{\small Universitat Polit\`ecnica de Catalunya }\\
{\small Barcelona (Catalonia) }\\
{\small {\tt \{lali,comellas,cdalfo,fiol\}@ma4.upc.edu}}}
\date{}

\begin{document}
\maketitle

\begin{abstract}
It has been shown that many networks associated with
complex systems are small-world (they have both a large local
clustering coefficient and a small diameter) and they are also
scale-free (the degrees are distributed according to a power law).
Moreover, these networks are very often hierarchical, as they
describe the modularity of the systems that are modeled. Most of
the studies for complex networks are based on stochastic methods.
However, a deterministic method, with an exact determination of the
main relevant parameters of the networks, has proven useful. Indeed,
this approach complements and enhances the probabilistic and
simulation techniques and, therefore, it provides a better
understanding of the systems modeled.
In this paper we find the radius, diameter, clustering coefficient
and degree distribution of a generic family of deterministic
hierarchical small-world scale-free networks that has been
considered for modeling real-life complex systems.
\end{abstract}





\section{Introduction}
With the publication in 1998 and 1999  of  the papers by Watts and
Strogatz on small-world networks~\cite{WaSt98} and by Barab\'asi and
Albert on scale-free networks~\cite{BaAl99}, there has been a
renewed interest in the study of networks associated to complex
systems that has received a considerable boost as an
interdisciplinary subject.

Many real-life networks,  transportation and communication systems
(including the power distribution  and telephone networks),
Internet~\cite{FaFaFa99},  World Wide Web~\cite{AlJeBa99},
and several social and biological
networks~\cite{JeToAlOlBa00,JeMaBaOl01,Ne01}, belong to a class of
networks known as small-world scale-free networks. All these
networks exhibit both a strong local clustering coefficient (nodes
have many mutual neighbors) and a small diameter. Another important
characteristic is that the number of links attached to the nodes
usually obeys a power law distribution (`scale-free' network).
Several authors also noticed that the modular structure of a network
can be characterized by a specific clustering distribution that
depends on the degree. The network is then called
hierarchical~\cite{RaSoMoOlBa02,SoVa04,WuRaBa03}. Moreover, with the
introduction of a new measuring technique for graphs, it has been
discovered that many real networks can also be categorized as
self-similar, see \cite{SoHaMa05}.

Along with these observational studies, researchers have developed
different models~\cite{AlBa02,DoMe02,Ne03}, most of them stochastic,
which should help to understand and predict the behavior and
characteristics of complex systems. However, new deterministic
models constructed by recursive methods, based on the existence of
`cliques' (clusters of nodes linked to each other), have also been
introduced \cite{BaRaVi01,CoFeRa04,DoGoMe02,JuKiKa02,ZhCoFeRo06}.
Such deterministic models have the advantage that they allow one to
analytically compute relevant properties and parameters, which may
be compared with data  from real and simulated networks. In
\cite{BaRaVi01}, Barab\'asi {\it et al.} proposed a simple
hierarchical family of deterministic networks and showed it had a
small-world scale-free nature. However, their null clustering coefficient
of all the vertices (the clustering coefficient of a vertex
is defined as the number of edges between the neighbors of this
vertex divided by the number of all possible edges between these
neighbors) contrasts with many real networks that have a high
clustering coefficient. Another family of hierarchical networks is
proposed in ~\cite{RaSoMoOlBa02}. It combines a modular structure
with a scale-free topology and models the metabolic networks of
living organisms and networks associated with generic system-level
cellular organizations. A simple variation of this hierarchical
network is considered in~\cite{RaBa03}, where other modular networks
(as WWW, the actor network, Internet at the domain level, etc.) are
studied. This model is further generalized in~\cite{No03}.

Several authors ~\cite{BaOl04,RaBa03,RaSoMoOlBa02} claim that a signature for a hierarchical network
on top of the small-world scale-free characteristics
is that the clustering of the vertices of the graph follows $C(k_i) \propto 1/k_i$, where $k_i$ is the degree of vertex $i$.

In this paper, we  study a family of hierarchical networks
recursively defined from an initial complete graph on $n$ vertices.
We find some of the main properties for this family:
radius, diameter, and degree and clustering distributions.

\section{The hierarchical graph  $H_{n,k}$}
In this section we generalize the constructions of deterministic
hierarchical graphs introduced by Ravasz {\em et
al.}~\cite{RaBa03,RaSoMoOlBa02} and Noh~\cite{No03}. Roughly
speaking, these graphs are constructed first by connecting a
selected root vertex of a complete graph $K_n$ to some vertices of
$n-1$ replicas of $K_n$, and establishing also some edges between
such copies of $K_n$. This gives a graph with $n^2$ vertices. Next,
$n-1$ replicas of the new whole structure are added, again with some
edges between them and to the same root vertex. At this step the
graph has $n^3$ vertices. Then we iterate the process until, for some integer $k\ge 1$,  the desired graph order $n^k$ is reached
(see below for a formal definition).
Our model enhances the modularity and self-similarity of the graph
obtained, and allows us to derive
exact expressions for the radius, diameter, degree and clustering
distributions.

\subsection{Definition, order and size}
Next we provide a recursive formal definition of the proposed family
of graphs, characterized by the parameters $n\ge 2$ (order of the
initial complete graph) and $k\ge 1$ (number of iterations or
dimension). This allows us to give also a direct definition and
derive an expression for the number of edges (the radius and the
diameter will be studied in the next section).


\begin{defi}
Let $n$ and $k$ be positive integers, $n\geq 2$. The hierarchical graph
$H_{n,k}$ has vertex set $V_{n,k}$, with $n^k$ vertices,
denoted by the $k$-tuples $x_1 x_2 x_3\ldots x_{k}$, $ x_i \in
\Z_{n}, 1\leq i\leq k$, and edge set $E_{n,k}$ defined recursively
as follows:
\begin{itemize}
\item
$H_{n,1}$ is the complete graph $K_{n}$.
\item
For  $k>1$,  $H_{n,k}$ is obtained from the union of $n$ copies of
$H_{n,k-1}$, each denoted by $H_{n,k-1}^{\alpha}$, $0\le \alpha\le n-1$,
and with vertices $x_2^{\alpha} x_3^{\alpha}\ldots
x_k^{\alpha}\equiv {\alpha} x_2 x_3\ldots x_{k}$, by adding the
following new edges $($where adjacencies are denoted by `$\sim$'$)$:
\begin{eqnarray}
\label{adj2'}
000\ldots 00  & \ \sim \ &  x_1 x_2 x_3\ldots x_{k-1}
x_{k}, \qquad    x_j\neq 0,\  1\leq j \leq k; \\
\label{adj3'}
x_{1}00\ldots 00 & \ \sim  \ &  y_{1}00\ldots00,
\qquad x_{1},y_{1}\neq 0,\ x_{1} \neq y_{1}.
\end{eqnarray}
\end{itemize}
\end{defi}

Alternatively, a direct definition of the edge set $E_{n,k}$ is
given by the following adjacency rules (when $i=0$, then $x_1 x_2
\ldots x_i$ is the empty string): \vskip-1cm
\begin{eqnarray}
\label{adj1}
x_1 x_2 \ldots x_k & \ \sim \ & x_1 x_2 \ldots x_{k-1} y_k, \quad y_k\neq x_k; \\
 \nonumber x_1 x_2 \ldots x_i 00 \ldots 0 &\ \sim \ & x_1 x_2
\ldots x_i x_{i+1} x_{i+2} \ldots x_k, \\
 & & \label{adj2} \hskip 2cm x_{j} \neq 0,\ i+1\leq j \leq k,\ 0\le i\le k-2;\\
 \nonumber  x_1 x_2 \ldots x_i 00 \ldots 0 &\ \sim \ & x_1 x_2
\ldots x_{i-1} y_i 00 \ldots 0, \\
  & & \label{adj3} \hskip 2cm  x_i,y_i\neq 0,\ y_i\neq x_i,\
1\le i\le k-1.
\end{eqnarray}
Notice that both conditions (\ref{adj2'}) and (\ref{adj3'}) of the
recursive definition correspond to (\ref{adj2}) with $i=0$, and
(\ref{adj3}) with $i=1$, respectively.

To illustrate our construction, Fig.~\ref{eldibu} shows the
hierarchical graphs $H_{4,k}$, for $k=1,2,3$. The following result
gives the number of edges of  $H_{n,k}$, which can be easily
computed by using the recursive definition.

\begin{center}
\begin{figure}[t]
\includegraphics[width=15cm]{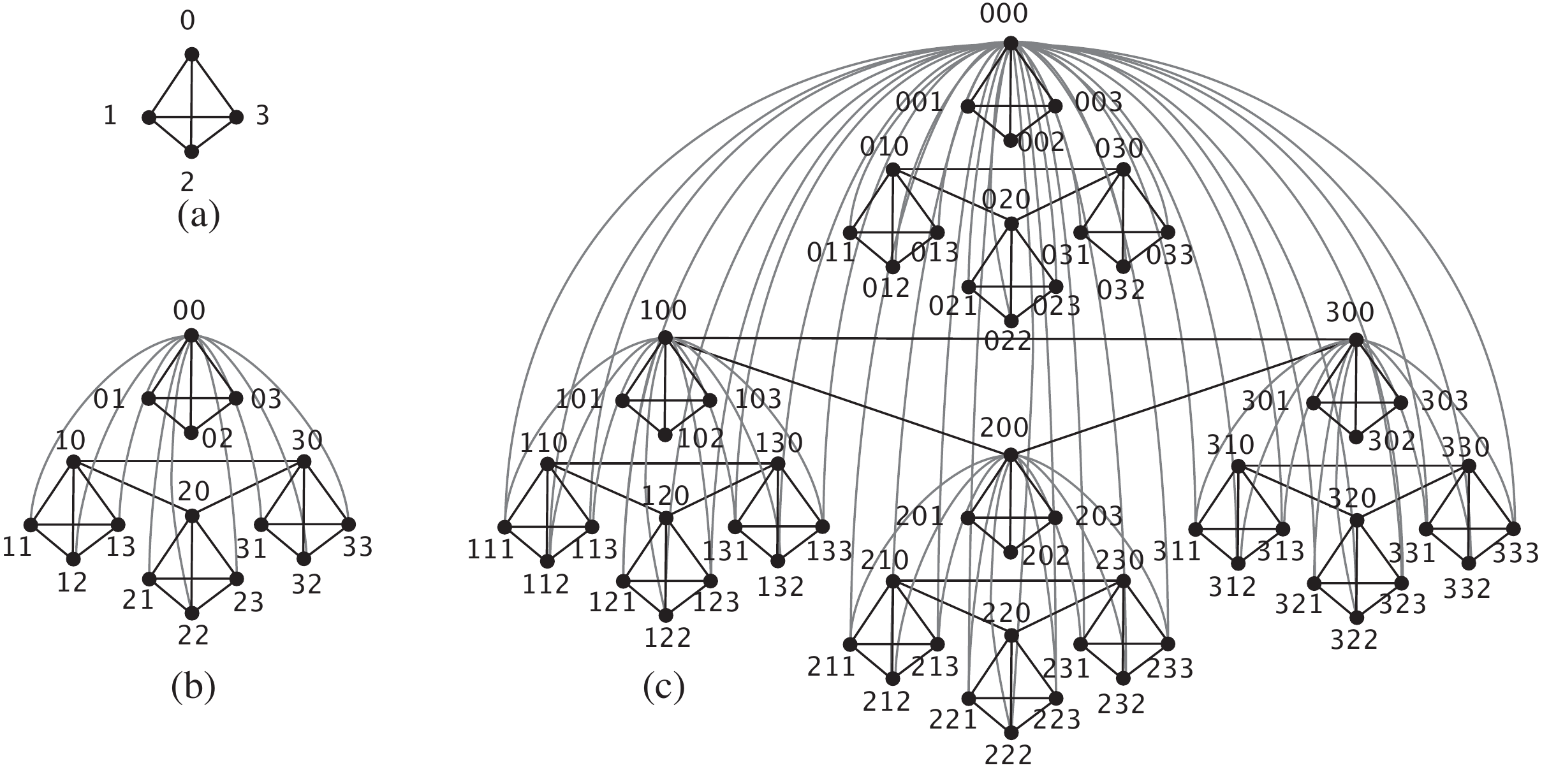}
\caption{Hierarchical graphs with initial order $4$: (a) $H_{4,1}$,
(b) $H_{4,2}$, (c) $H_{4,3}$. } \label{eldibu}
\end{figure}
\end{center}

\begin{pr}
\label{size} The size of $H_{n,k}$ is
\begin{equation}
\label{num-edges}
\mathcal{j}E_{n,k}\mathcal{j}=\frac{3}{2}n^{k+1}-(n-1)^{k+1}-2n^k-\frac{n}{2}+1.
\end{equation}
\end{pr}
\begin{proof}
When constructing $H_{n,k}$ from $n$ copies of $H_{n-1,k}$, the
adjacencies (\ref{adj2'}) and (\ref{adj3'}) introduce $(n-1)^k$  and
${n-1\choose 2}$ new edges, respectively. Therefore,
$$
|E_{n,k}|=n |E_{n,k-1}|+(n-1)^k + {n-1\choose 2}. $$ By applying
recursively this formula and taking into account that
$|E_{n,1}|={n\choose 2}$, we get
\begin{equation}
\label{num-edges} |E_{n,k}|=n^{k-1} {n\choose 2} +
\sum_{i=2}^{k}n^{k-i}(n-1)^i + {n-1\choose 2}\sum_{i=0}^{k-2}n^i,
\end{equation}
which yields the result.
\end{proof}

\subsection{Hierarchical properties}
\label{hierarquicalprop}

The hierarchical properties of the graphs $H_{n,k}$ are summarized
by the following facts, that are a direct consequences of the
definition:
\begin{itemize}
\item[$(a)$]
According to (\ref{adj1}), for each sequence of fixed values
$\alpha_i\in \Z_n$, $1\leq i\leq k-1$, the vertex set  $\{\alpha_1
\alpha_2\ldots \alpha_{k-1}x_k :
x_k\in\Z_n\}$ induces a subgraph isomorphic to $K_n$. 
\medskip
\item[$(b)$]
Vertex $\r:=00\ldots 0$, which we distinguish and call \emph{root},
is adjacent by (\ref{adj2}) to vertices $x_1 x_2\ldots x_{k}$,
$x_i\neq 0$, for all $1\leq i \leq k $, which we call
\emph{peripheral}.
\medskip
\item[$(c)$]
For every $i$, $1\le i \le k-1$, $H_{n,k}$ can be decomposed into
$n^i$ vertex-disjoint subgraphs isomorphic to $H_{n,k-i}$. Each of
such (induced) subgraphs is denoted by $H_{n,k-i}^{\mbox{\scriptsize
$\vecalpha$}}$ and has vertex labels  $\vecalpha
x_{i+1}x_{i+2}\ldots x_k$, with $\vecalpha=\alpha_1\alpha_2 \ldots
\alpha_i \in \Z_n^i$ being a fixed sequence. In particular, for
$i=1$, $H_{n,k}$ has $n$ subgraphs $H_{n,k-1}^\alpha$,
$\alpha=0,1,\ldots, n-1$, as stated in the recursive definition.
\medskip
\item[$(d)$]
The root vertex of the subgraph $H_{n,k-i}^{\mbox{\scriptsize
$\vecalpha$}}$ is $\vecalpha \underbrace{00\ldots0}_{k-i}$. Thus,
the total number of root vertices of all the subgraphs, including
the one in $H_{n,k}$, is
\begin{equation}
\label{num-roots}
 1+(n-1)\sum_{i=1}^{k-1} n^{i-1}=n^{k-1},
\end{equation}
as expected since a given vertex $x_1x_2\ldots x_k$ is a root (of
some subgraph) if and only if $x_k=0$.
\item[$(e)$]
The peripheral vertices of the subgraph
$H_{n,k-i}^{\mbox{\scriptsize $\vecalpha$}}$ are of the form
$\vecalpha x_{i+1}x_{i+2}\ldots x_{k}$, where $x_j\neq 0$, $i+1\le
j\le k$. Thus, the total number of peripheral vertices of all the
subgraphs, including those in $H_{n,k}$, see $(b)$, is
\begin{equation}
\label{num-peripheral}
(n-1)^k+(n-1)\sum_{i=1}^{k-1}
n^{i-1}(n-1)^{k-i}=n^{k-1}(n-1),
\end{equation}
as expected since  $x_1x_2\ldots x_k$ is a peripheral vertex (of
some subgraph) if and only if $x_k\neq 0$. Note that, adding up
(\ref{num-roots}) and (\ref{num-peripheral}), we get $n^k=|V_{n,k}|$,
so that every vertex of $H_{n,k}$ is a root or peripheral of some subgraph
isomorphic to $H_{n,k'}$, $1\le k'\le k$.
\medskip
\item[$(f)$]
By collapsing in $H_{n,k}$ each of the  $n^i$  subgraphs
$H_{n,k-i}^{\mbox{\scriptsize $\vecalpha$}}$, $\vecalpha\in\Z_n^i$,
into a single vertex and all multiple edges into one, we obtain a
graph isomorphic to $H_{n,i}$.
\medskip
\item[$(g)$]
According to (\ref{adj3}), for every fixed $i$, $1\le i\le k$, and
given a sequence $\vecalpha\in\Z_n^{i-1}$, there exist all possible
edges among the $n-1$ vertices labeled $\vecalpha x_i00\ldots0$ with
$x_i \in \Z_n^*=\{1,2,\ldots,n-1\}$, that is, the root vertices of
$H_{n,k-i}^{\mbox{\scriptsize $\vecalpha$} x_i}$. Thus, these edges
induce a complete graph isomorphic to $K_{n-1}$.
\end{itemize}

\section{Radius and Diameter}

In this section we determine the radius and diameter of $H_{n,k}$
by using a recursive method. With this aim,
let us first introduce some notation concerning $H_{n,k}$. Let
$\dist_k(\x,\y)$ denote the distance between vertices $\x,\y\in
V_{n,k}$ in $H_{n,k}$; and $\dist_k(\x,U):=\min_{\u\in
U}\{\dist_k(\x,\u)\}$. Let $\r^{\alpha}=\alpha 00\ldots0$ be the root
vertex of $H_{n,k-1}^{\alpha}$, $\alpha\in \Z_n$ (as stated before,
$\r$ stands for the root vertex of $H_{n,k}$). Let $P$ and
$P^\alpha$, $\alpha\in \Z_n$, denote the set of peripheral vertices
of $H_{n,k}$ and $H_{n,k-1}^\alpha$, respectively.

\begin{pr}
Let $r_k,\ecc_k(\r),D_k$ denote, respectively, the radius, the
eccentricity of the root $\r$, and the diameter of $H_{n,k}$. Then,
\begin{enumerate}
\item[$(a)$]
$r_k=\ecc_k(\r)=k$.
\item[$(b)$]
$D_k=2k-1$.
\end{enumerate}
\end{pr}
\begin{proof}
\begin{enumerate}
\item[$(a)$]
The radius of $H_{n,k}$ coincides with the eccentricity of the root:
$r_k=\ecc_k(\r)=k$.
\item[$(b)$]
By induction on $k$.
\\
For $k=1$: As $H_{n,1}=K_n$, then $D_1=1$.
\\
Assume that, for some fixed $k>1$, $D_k=2k-1$.
\\
Then, for $k'=k+1$: As $H_{n,k'}$ is made from $n$ copies of $H_{n,k}$
(called copy 0, copy 1,\ldots, copy $n-1$), two further vertices
in $H_{n,k'}$ must be in different copies of $H_{n,k}$. If none of
these two vertices is in the copy $0$ of $H_{n,k}$, then both
copies are joined by their roots. Then, the diameter of $H_{n,k'}$
is:
$$
D_{k'} = \ecc_k(\r)+\ecc_k(\r)+1 = 2k+1 = 2k'-1,
$$
where $\r$ is the root of any of the two copies of $H_{n,k}$. On
the other hand, if one of the two vertices is  in the copy $0$ of $H_{n,k}$,
then both copies are joined from the root of the copy 0 to the
peripheral vertices of the other copy of $H_{n,k}$. Then, the
diameter of $H_{n,k'}$ is:
$$
D_{k'} = \ecc_k(\r)+\ecc_k(\p)+1 = 2k+1 = 2k'-1,
$$
where $\p$ is one of the peripheral vertices of the non-zero copy of
$H_{n,k}$, and $\ecc_k(\p)=k$.
\end{enumerate}
\end{proof}

Then, from the result on the diameter and property $(c)$ in
Subsection \ref{hierarquicalprop}, we have that the distance between
two vertices $\x$ and $\y$  of $H_{n,k}$, with maximum common prefix of
length $i=|\x\cap\y|$, satisfies
$$
\dist(\x,\y)\le 2(k-i)-1.
$$

Alternatively, we can give recursive proofs of these results.
Indeed, let us consider the case of the diameter. With this aim, we
first give the following result that follows from the recursive
definition of $H_{n,k}$:

\begin{lemma}
\label{property} Let  $\x$ and $\y$ be two vertices in $H_{n,k}$, $k>1$.
Then, depending on the subgraphs $H_{n,k-1}$  where such vertices
belong to, we are in one of the following three cases:
\begin{enumerate}
\item[$(a)$]
If $\x,\y\in V_{n,k-1}^{\alpha}$ for some $\alpha\in \Z_{n}$, that is,
$\x=\alpha\x'$ and $\y=\alpha\y'$, then,
$$
\dist_k(\x,\y)=\dist_{k-1}(\x',\y').
$$
\item[$(b)$]
If $\x\in V_{n,k-1}^0$ and $\y\in V_{n,k-1}^{\alpha}$ for some
$\alpha\in \Z_{n}^*$, that is $\x=0\x'$, $\y=\alpha\y'$, with
$\alpha\not=0$, then,
$$
\dist_k(\x,\y)=\dist_{k-1}(\x',\r^0)+1+\dist_{k-1}(\y',P^{\alpha}).
$$
\item[$(c)$]
If $\x\in V_{n,k-1}^{\alpha}$ and $\y\in V_{n,k-1}^{\beta}$ for some
$\alpha,\beta\in \Z_{n}^*$, $\alpha\neq \beta$, that is
$\x=\alpha\x'$, $\y=\beta\y'$, with $\alpha,\beta\not=0$, then,
$$
\dist_k(\x,\y)=\min\{
\dist_{k-1}(\x',P^\alpha)+2+\dist_{k-1}(\y',P^\beta),\,\dist_{k-1}(\x',\r^\alpha)+1+\dist_{k-1}(\r^\beta,\y')\}.
$$
\end{enumerate}
\end{lemma}

\begin{lemma}
\label{lemma:diam-low} For any vertex $\x$ in $H_{n,k}$ we have:
$$
\dist_k(\x,\r)\le \left\{
\begin{array}{ll}
k-1 &  \mbox{\ \ if\ \ } \x=0\x',\\
k & \textrm{\ \  otherwise,}
\end{array}
\right. \quad \mbox{ and } \qquad \dist_k(\x,P)\le \left\{
\begin{array}{ll}
k & \mbox{\ \ if\ \ } \x=0\x',\\
k-1 & \textrm{\ \  otherwise.}
\end{array}
\right.
$$
\end{lemma}

\begin{proof}
By induction on $k$.
\\
Case $k=1$: If $\x=\mbox{\boldmath{$0$}}=\r$, then $\dist_1(\x,\r^0)=0$ and
$\dist_1(\x,P)=1$. Otherwise, $\x\in P=\Z_n^*$, and then
$\dist_1(\x,\r)=1$ and $\dist_1(\x,P)=0$.
\\
Case $k>1$: We observe that, from the recursive definition of
$H_{n,k}$,
$$
\dist_k(\x,\r)= \left\{
\begin{array}{ll}
\dist_{k-1}(\x',\r^0) & \textrm{\ \ if\ \ } \x=0\x',\\
\dist_{k-1}(\x',P^\alpha)+1  & \textrm{\ \ if $\x=\alpha\x'$ and
$\alpha\neq 0$},
\end{array}
\right.
$$
and
$$
\dist_k(\x,P)= \left\{
\begin{array}{ll}
\dist_{k-1}(\x',\r^0)+1& \textrm{\ \  if\ \  } \x=0\x',\\
\dist_{k-1}(\x',P^\alpha) & \textrm{\ \  if $\x=\alpha\x'$ and
$\alpha\neq 0$}.
\end{array}
\right.
$$
Then, by the induction hypothesis, the lemma holds.
\end{proof}

In the next result, $\z^{01}=0101\dots$ and $\z^{10}=1010\dots$
denote any vertex  $x_1x_2\ldots x_i\ldots$ of $H_{n,k}$ or
$H_{n,k-1}$, where $x_i\equiv i+1$ $(\mod 2)$ and $x_i\equiv i$
$(\mod 2)$, respectively.

\begin{lemma}
\label{lemma:diam-up} In $H_{n,k}$, the following equalities hold:
\begin{enumerate}
\item[$(a)$]
$\dist_k(\z^{01}, \r)  =  \dist_k(\z^{10},P) = k-1$,
\item[$(b)$]
$\dist_k(\z^{10},\r)  =  \dist_k(\z^{01}, P) =  k$.
\end{enumerate}
\end{lemma}

\begin{proof}
By induction on $k$.
\\
Case $k=1$: $H_{n,k}$  is the complete graph $K_n$, and the result
clearly holds.
\\
Case $k>1$: From Lemma~\ref{property} we have:
\begin{enumerate}
\item[$(a)$] $\dist_k(\z^{01},\r)=\dist_{k-1}(\z^{10},\r^0)=k-1$,
\item[]
$\dist_k(\z^{10},P)=\dist_{k-1}(\z^{01}, P^0)=k-1$;

\item[$(b)$] $\dist_k(\z^{10},\r)=\dist_{k-1}(\z^{01}, P^1)+1=k$,
\item[]
$\dist_k(\z^{01},P)=\dist_{k-1}(\z^{10}, \r^0)+1=k-1+1=k$.
\end{enumerate}
\end{proof}

Now we can give the result about the diameter of $H_{n,k}$.
\begin{pr}
\label{prop:diameter}
The diameter of $H_{n,k}$ is $D_k=2k-1$.
\end{pr}
\begin{proof}
First we prove by induction on $k$ that, for any given pair of
vertices of $H_{n,k}$, $\x$ and $\y$, we have $\dist_k(\x,\y)\le
2k-1$.
\\
Case $k=1$: The result trivially holds since $H_{n,1}=K_{n}$ and
$D_1=1$.
\\
Case $k>1$: Considering the three cases of Lemma~\ref{property} and
by using the induction hypothesis, we have:
\begin{itemize}
\item[$(a)$]
If $\x,\y\in V_{n,k-1}^{\alpha}$ for some $\alpha\in \Z_{n}$, that is,
$\x=\alpha\x'$ and $\y=\alpha\y'$, then,
$$
\dist_k(\x,\y)=\dist_{k-1}(\x',\y')\le 2(k-1)-1=2k-3<2k-1.
$$
\item[$(b)$]
If $\x\in V_{n,k-1}^0$ and $\y\in V_{n,k-1}^{\alpha}$ for some
$\alpha\in \Z_{n}^*$, that is $\x=0\x'$, $\y=\alpha\y'$, with
$\alpha\not=0$, then,
$$
\dist_k(\x,\y)=\dist_{k-1}(\x',\r^0)+1+\dist_{k-1}(\y',P^{\alpha})\le
2(k-1)+1=2k-1,
$$
since, by Lemma~\ref{lemma:diam-low}, $\dist_{k-1}(\x',\r^0)\le k-1$
and $\dist_{k-1}(\y',P^{\alpha})\le k-1$.
\item[$(c)$]
If $\x\in V_{n,k-1}^{\alpha}$ and $\y\in V_{n,k-1}^{\beta}$ for some
$\alpha,\beta\in \Z_{n}^*$, $\alpha\neq \beta$, that is
$\x=\alpha\x'$, $\y=\beta\y'$, with $\alpha,\beta\not=0$, then
\begin{eqnarray*}\hskip -.8cm
\dist_k(\x,\y) & = & \min\{
\dist_{k-1}(\x',P^\alpha)+2+\dist_{k-1}(\y',P^\beta),\,\dist_{k-1}(\x',\r^\alpha)+1+\dist_{k-1}(\r^\beta,\y')\}\\
 & \le & 2(k-1)+1=2k-1,
\end{eqnarray*}
since, by Lemma~\ref{lemma:diam-up}, $\dist_{k-1}(\x',\r^\alpha)\le
k-1$ and $\dist_{k-1}(\r^\beta,\y')\le k-1$.
\end{itemize}

Now, we have to prove that there exist two vertices in $H_{n,k}$ at
distance exactly $2k-1$.  Let $\x=\z^{01}$ and $\y=\z^{10}$. It
follows from Lemmas~\ref{property} and~\ref{lemma:diam-up} that
$\dist_k(\x,\y)=2k-1$. This completes
the proof.
\end{proof}

Note that the diameter scales logarithmically  with the order
$N=|V_{n,k}|=n^k$, since 
$D_k=\frac{2}{\log n}\log N -1$.
This property, together with the high value of the clustering coefficient
(see next section), shows that this is a small-world network.

\section{Degree and clustering distribution}
In this section we study the degree and clustering distributions of
the graph $H_{n,k}$. 
\begin{pr}
\label{propo-degree}
 The  vertex degree distribution in $H_{n,k}$ is
as follows:
\begin{itemize}
\item[$(a)$]
The root vertex $\r$ of  $H_{n,k}$ has degree
$$
\delta(\r)=\frac{(n-1)^{k+1}-(n-1)}{n-2}.$$
\item[$(b)$]
The degree of the root vertex $\r_{k-i}^{\mbox{\scriptsize
$\vecalpha$}}$ of each of the $(n-1)n^{i-1}$ subgraphs
$H_{n,k-i}^{\mbox{\scriptsize $\vecalpha$}}$, with
$i=1,2,\dots,k-1$, $\vecalpha=\alpha_1\alpha_2\ldots\alpha_i\in
\Z_n^i$ and  $\alpha_i\neq 0$, is
$$
\delta(\r_{k-i}^{\mbox{\scriptsize
$\vecalpha$}})=\frac{(n-1)^{k-i+1}-(n-1)}{n-2}+(n-2).
$$
\item[$(c)$]
The  degree of the $(n-1)^k$ peripheral vertices $\p$ of $H_{n,k}$ is
$$ \delta(\p)=n+k-2.
$$
\item[$(d)$]
The  degree of the $(n-1)^{k-i}n^{i-1}$ peripheral vertices
$\p_{k-i}^{\mbox{\scriptsize $\vecalpha$}}$ of the subgraphs
$H_{n,k-i}^{\mbox{\scriptsize $\vecalpha$}}$, with
$i=1,2,\dots,k-1$, $\vecalpha=\alpha_1\alpha_2\ldots\alpha_i\in
\Z_n^i$ and  $\alpha_i\neq 0$,  is
$$
\delta(\p_{k-i}^{\mbox{\scriptsize $\vecalpha$}})= n+k-i-2.
$$
\end{itemize}
\end{pr}
\begin{proof}
$(a)$ By the adjacency conditions (\ref{adj1}) and (\ref{adj2}), the
root of $H_{n,k}$ has degree
$$
\textstyle
 \delta(\r)=\sum_{i=1}^k(n-1)^i
=\frac{(n-1)^{k+1}-(n-1)}{n-2}.
$$
$(b)$ The root of the subgraph $H_{n,k-i}^{\mbox{\scriptsize
$\vecalpha$}}$, $i=1,2,\dots,k-1$,
$\vecalpha=\alpha_1\alpha_2\ldots\alpha_i\in \Z_n^i$ and
$\alpha_i\neq 0$, is adjacent, by $(a)$, to
$\frac{(n-1)^{k-i+1}-n+1}{n-2}$ vertices belonging to the same
subgraph, and also, by (\ref{adj3}), to the $n-2$
other roots `at the same level'.

\noindent $(c)$ Each peripheral vertex of $H_{n,k}$ is adjacent, by
(\ref{adj1}), to $n-1$ vertices and, by (\ref{adj2}), to $k-1$
 roots of other subgraphs.

\noindent $(d)$ Each peripheral vertex of
$H_{n,k-i}^{\mbox{\scriptsize $\vecalpha$}}$, $i=1,2,\dots,k-1$,
$\vecalpha=\alpha_1\alpha_2\ldots\alpha_i\in \Z_n^i$ and
$\alpha_i\neq 0$, is adjacent, by (\ref{adj1}), to $n-1$ vertices
(of the subgraph isomorphic to $K_n$) and, by (\ref{adj2}), to $k-i$
roots of other subgraphs.
\end{proof}

The above results on the degree distribution of $H_{n,k}$ are
summarized in Table~\ref{taula:grau}. Note that, from such a
distribution, we can obtain again Proposition \ref{size} since the
number of edges can be computed from
$$
2|E_{n,k}|=\delta(\r)+\sum_{i=1}^{k-1}(n-1)n^{i-1}\delta(\r_{k-i}^{\mbox{\scriptsize
$\vecalpha$}})+(n-1)^k\delta(\p)+\sum_{i=1}^{k-1}(n-1)^{k-i}n^{i-1}\delta(\p_{k-i}^{\mbox{\scriptsize
$\vecalpha$}}),
$$
which yields (\ref{num-edges}). Moreover,  using this result, we see that, for a large dimension $k$, the average degree turns out to be of order
$$
\overline{\delta}=\frac{2|E_{n,k}|}{|V_{n,k}|}= \frac{3 n^{k+1}-4
n^k-2(n-1)^{k+1}-n+2}{n^k}\sim n+2k-2.
$$


From the degree distribution and for large $k$ we see that the number
of vertices with a given degree $z$, $N_{n,k}(z)$, decreases as a power
of the degree $z$ and, therefore, the graph is scale-free~\cite{BaAl99,CoFeRa04,DoMe02}.
As the degree distribution of the graph is discrete,
to relate the exponent of this discrete degree
distribution to the standard
$\gamma$ exponent of a continuous degree distribution
for random scale free networks, we use a cumulative distribution
$$
P_{\mbox{\footnotesize{cum}}}(z)\equiv\sum_{z'\ge
z}{\mathcal{j}N_{n,k}(z')\mathcal{j}}/
{\mathcal{j}V_{n,k}\mathcal{j}}\sim z^{1-\gamma},
$$
where $z$ and
$z'$ are points of the discrete degree spectrum. When $z =
\frac{(n-1)^{k-i+1}-n+2}{n-2}$, there are exactly $(n-1)n^{i-1}$
vertices with degree $z$. The number of vertices with this or a
higher degree is
$$
(n-1)n^{i-1}+\cdots+(n-1)n+(n-1)+1=1+
(n-1)\sum_{j=0}^{i-1}n^{j}=n^i.
$$
Then, we have
$z^{1-\gamma}={n^i}/{n^k}=n^{i-k}.$ Therefore, for large $k$,
$((n-1)^{k-i})^{1-\gamma} \sim n^{i-k}$ and
$$\gamma \sim 1+\frac{\log n}{\log (n-1)}.
$$
For $n=5$ this gives the same value of $\gamma$ as in the case of
the hierarchical network introduced in~\cite{RaBa03}. This network
can be obtained from $H_{5,k}$ by deleting the edges that join the
roots of $H_{5,k-i}^j$, $j\neq 0, 1\leq i\leq k-2$.

\begin{table}[htbp]
\caption{Degree and clustering distribution for $H_{n,k}$. }
\label{taula:grau}
\begin{tabular}{p{3.2cm} c c c}
\hline\hline
Vertex class &  No. vertices  &      Degree                  &       Clustering coefficient   \\
\hline
$H_{n,k}$ root    &  1  & $\frac{(n-1)^{k+1}-(n-1)}{n-2}$ & $\frac{(n-2)^2}{(n-1)^{k+1}-2n+3}$ \\
&  &  & \\
$H_{n,k-i}^{\mbox{\scriptsize $\vecalpha$}}$ roots & $(n-1)n^{i-1}$
& $\frac{(n-1)^{k-i+1}-(n-1)}{n-2}+n-2$
& $\frac{(n-2)^2}{(n-1)^{k-i+1}+(n-1)^2-3n+4}$ \\
\footnotesize{$i\!\!=\!\!1,2,\ldots,k-1$, \hskip.5cm
$\vecalpha\!\!=\!\!\alpha_1\alpha_2\ldots\alpha_i\!\!\in\!\!\Z_n^i$, $\alpha_i\neq0$} & &  \\
&  &  & \\
$H_{n,k}$ peripheral & $(n-1)^k$ & $n+k-2$ & $\frac{(n-1)^2+(2k-3)(n-1)+2-2k}{(n+k-2)(n+k-3)}$ \\
&  &  & \\
$H_{n,k-i}^{\mbox{\scriptsize $\vecalpha$}}$ peripheral & $(n-1)^{k-i}n^{i-1}$ & $n+k-i-2$
& $\frac{(n-1)^2+(2k-2i-3)(n-1)+2+2i-2k}{(n+k-i-2)(n+k-i-3)}$ \\
\footnotesize{$i\!\!=\!\!1,2,\ldots,k-1$, \hskip.5cm
$\vecalpha\!\!=\!\!\alpha_1\alpha_2\ldots\alpha_i\!\!\in\!\!\Z_n^i$, $\alpha_i\neq0$} & & \\
\hline\hline
\end{tabular}
\end{table}

Next we find the clustering distribution of the vertices of
$H_{n,k}$. The clustering coefficient of a graph $G$ measures its
`connectedness' and is another parameter used to characterize
small-world and scale-free networks. The clustering coefficient of a
vertex was introduced in~\cite{WaSt98} to quantify this concept. For
each vertex $v\in V(G)$ with degree $\dg_v$, its {\it clustering
coefficient\/} $c(v)$ is defined as the fraction of the
${\dg_v\choose 2}$ possible edges among the neighbors of $v$ that
are present in $G$. More precisely, if $\epsilon_v$ is the number of
edges between the $\dg_v$ vertices adjacent to vertex $v$, its
clustering coefficient is
\begin{equation}
\label{c(v)}
c(v)=\frac{2\epsilon_v}{\dg_v(\dg_v-1)},
\end{equation}
whereas the {\em clustering coefficient} of $G$, denoted by $c(G)$,
is the average of $c(v)$ over all nodes $v$ of $G$:
\begin{equation}
\label{c(G)} c(G)=\frac{1}{|V(G)|}\sum_{v\in V(G)}c(v).
\end{equation}

Another definition of {\em clustering coefficient} of $G$ was
given in \cite{NeWaSt02} as
\begin{equation}
\label{c'(G)} c'(G)=\frac{3\,T(G)}{\tau(G)},
\end{equation}
where $\tau (G)$ and $T(G)$ are, respectively, the number of {\it
triangles} (subgraphs isomorphic to $K_3$) and the number of {\it
triples} (subgraphs isomorphic to a path on $3$ vertices) of $G$.
A triple at a vertex $v$ is a $3$-path with central vertex $v$. Thus
the number of triples at $v$ is
\begin{equation}
\label{tau(v)}
 \tau(v)={\dg_v \choose 2}=\frac{\dg_v(\dg_v-1)}{2}.
\end{equation}
The total number of triples of $G$ is denoted by $\tau(G)=\sum_{v\in
V(G)}\tau(v)$. Using these parameters, note that the clustering
coefficient of a vertex $v$ can also be written as
$c(v)=\frac{T(v)}{\tau(v)}$, where $T(v)={\delta_v \choose 2}$ is
the number of triangles of $G$ that contain the vertex $v$.
From this result,
we get
that $c(G)=c'(G)$ if, and only if,
$$
|V(G)|=\frac{\sum_{v\in V(G)}\tau(v)}{\sum_{v\in
V(G)}T(v)}\sum_{v\in V(G)} \frac{T(v)}{\tau(v)}.
$$
This is true for regular graphs or for graphs such that all their vertices
have the same clustering coefficient. In fact, $c'(G)$ was already
known in the context of social networks as {\em transitivity
coefficient}.

We first compute the clustering coefficient and, then, the
transitivity coefficient.

\begin{pr}
The clustering distribution of $H_{n,k}$ is the following:
\begin{itemize}
\item[$(a)$]
The root $\r$ of $H_{n,k}$ has clustering coefficient
$$
c(\r)=\frac{(n-2)^2}{(n-1)^{k+1}-2n+3}.
$$
\item[$(b)$]
The clustering coefficient of the root vertex
$\r_{k-i}^{\mbox{\scriptsize $\vecalpha$}}$ of each of the
$(n-1)n^{i-1}$ subgraphs $H_{n,k-i}^{\mbox{\scriptsize
$\vecalpha$}}$, with $i=1,2,\dots,k-1$,
$\vecalpha=\alpha_1\alpha_2\ldots\alpha_i\in \Z_n^i$ and
$\alpha_i\neq 0$, is
$$
c(\r_{k-i}^{\mbox{\scriptsize
$\vecalpha$}})=\frac{(n-2)^2}{(n-1)n^{k-i+1}+(n-1)^2-3n+4}.
$$
\item[$(c)$]
The  clustering coefficient of the $(n-1)^k$ peripheral vertices
$\p$ of $H_{n,k}$ is
$$
c(\p)=\frac{(n-1)^2+(2k-3)(n-1)+2-2k}{(n+k-2)(n+k-3)}.
$$
\item[$(d)$]
The  clustering coefficient of the $(n-1)^{k-i}n^{i-1}$ peripheral
vertices $\p_{k-i}^{\mbox{\scriptsize $\vecalpha$}}$ of the
subgraphs $H_{n,k-i}^{\mbox{\scriptsize $\vecalpha$}}$, with
$i=1,2,\dots,k-1$, $\vecalpha=\alpha_1\alpha_2\ldots\alpha_i\in
\Z_n^i$ and  $\alpha_i\neq 0$ is
$$
c(\p_{k-i}^{\mbox{\scriptsize
$\vecalpha$}})=\frac{(n-1)^2+(2k-2i-3)(n-1)+2+2i-2k}{(n+k-i-2)(n+k-i-3)}.
$$
\end{itemize}
\end{pr}
\begin{proof}
We prove only three of the cases, as the proof of the other is similar.
\begin{itemize}
\item[$(a)$]
As the root of  $H_{n,k}$ is adjacent to $\sum_{i=1}^k(n-1)^i$
 vertices  with degree $n-2$, its clustering coefficient is
$$
c(\r)=\frac{\frac{n-2}{2}\frac{(n-1)^{k+1}-n+1}{n-2}}
{\frac{1}{2}\frac{(n-1)^{k+1}-n+1}{n-2}\left(\frac{(n-1)^{k+1}-n+1}{n-2}-1\right)}=
\frac{(n-2)^2}{(n-1)^{k+1}-2n+3}.
$$
\item[$(b)$]
The roots of $H_{n,k-i}^{\mbox{\scriptsize $\vecalpha$}}$
($i=1,2,\dots,k-1$, $\alpha_i\neq 0$) have clustering coefficient
\begin{eqnarray*}
c(\r_{k-i}^{\mbox{\scriptsize $\vecalpha$}}) & = &
\frac{\frac{n-2}{2}\frac{(n-1)^{k-i+1}-n+1}{n-2}+\frac{(n-2)(n-3)}{2}}{\frac{1}{2}
\left(\frac{(n-1)^{k-i+1}-n+1}{n-2}+n-2\right)\left(\frac{(n-1)^{k-i+1}-n+1}{n-2}+n-3\right)}\\
& = &\frac{(n-2)^2}{(n-1)^{k-i+1}+(n-1)^2-3n+4}.
\end{eqnarray*}
\item[$(d)$]  The clustering coefficient of the peripheral vertices of $H_{n,k-i}^{\mbox{\scriptsize $\vecalpha$}}$
($i=1,2,\dots,k-1$, $\alpha_i\neq 0$)  is
\begin{eqnarray*}
c(\p_{k-i}^{\mbox{\scriptsize
$\vecalpha$}}) & = &\frac{\frac{(n-1)(n-2)}{2}+(n-2)(k-i-1)}{\frac{1}{2}(n+k-i-2)(n+k-i-3)}\\
& = & \frac{(n-1)^2+(2k-2i-3)(n-1)+2+2i-2k}{(n+k-i-2)(n+k-i-3)}.
\end{eqnarray*}
In particular, note that, for  $i=k-1$,  the peripheral vertices of
$H_{n,1}^\alpha$, $\alpha\neq 0$, have clustering coefficient
$\frac{(n-1)^2-n+1}{(n-1)n}=1$.
\end{itemize}
\end{proof}
\begin{figure}[htbp]
\begin{center}
\includegraphics[width=12cm]{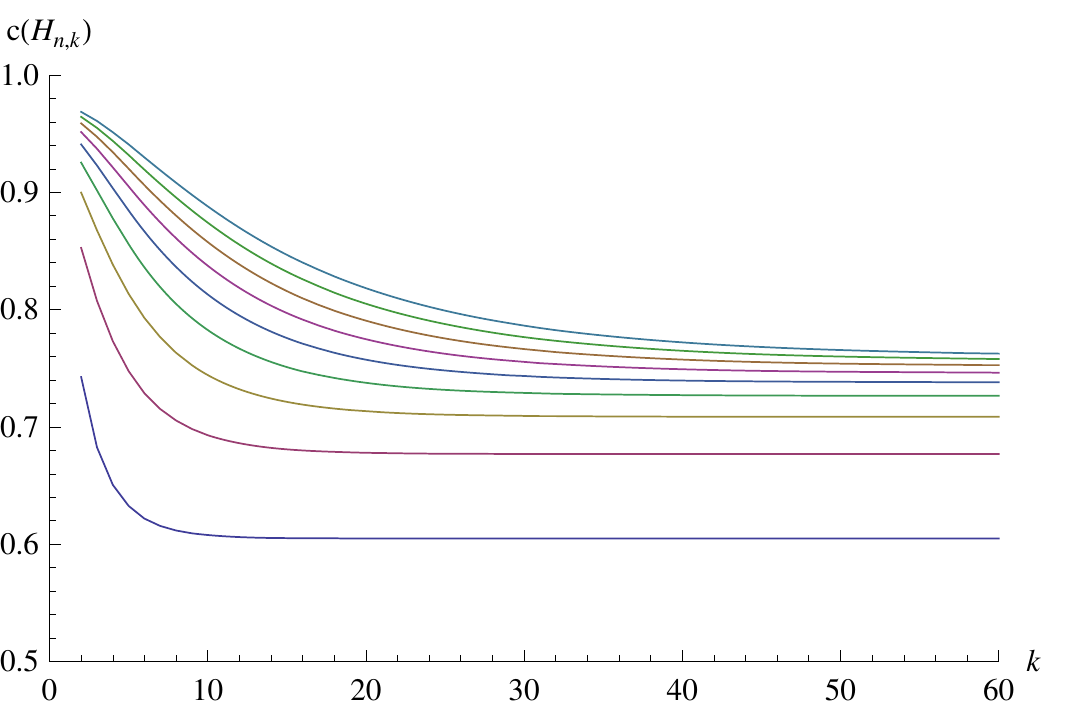}
\caption{The clustering coefficient of $H_{n,k}$ for $n=4,6,\ldots,
20$.} \label{clustdib}
\end{center}
\end{figure}

The above results on the clustering distribution are summarized in
Table~\ref{taula:grau}. From these results, we can compute the
clustering coefficient of $H_{n,k}$, which is shown in
Fig.~\ref{clustdib}. The clustering coefficient tends to 1 for large
$n$.


We think that this constant value for the clustering coefficient, which is independent of the order of the graph, together with the
$\gamma$ value of the power law distribution of the degrees, is also
a good characterization of {\em modular} hierarchical networks.
Observations in metabolic networks of different organisms show that
they are highly modular and have these properties,  confirming the
claim, see~\cite{BaOl04,RaSoMoOlBa02}.


To find the transitivity coefficient, we need to calculate the number
of triangles and the number of triples of the graph.

\begin{pr}
The number $T_{n,k}$ of triangles of $H_{n,k}$ is
$$
T_{n,k}=\frac{1}{2}(n-2)\left(1-\frac{n}{3}-(n-1)^{k+1}+\frac{2}{3}n^k(2n-3)\right).
$$
\end{pr}
\begin{proof}
When constructing $H_{n,k}$ from $n$ copies of $H_{n,k-1}$, the
adjacencies (\ref{adj2'}) and (\ref{adj3'}) introduce $(n -
1)^{k-1}{n-1\choose 2}$  and ${n-1\choose 3}$ new triangles,
respectively. Therefore,
$$
T_{n,k} =n T_{n,k-1}+(n - 1)^{k-1}{n-1\choose 2} + {n-1\choose 3}.
$$

By applying recursively this formula and taking into account that
$T_{n,1}={n\choose 3}$, we get the result.
\end{proof}

Moreover, from the results of Proposition \ref{propo-degree} (or
Table 1) giving the number of vertices of each degree, we have the
following result for the number of triples (we omit the obtained
explicit formula, because of its length):

\begin{pr}
The number $\tau_{n,k}$ of triples of $H_{n,k}$ is
$$
\textstyle \tau_{n,k} =  {\delta(\r)\choose 2}+
(n-1)\sum_{i=1}^{k-1}n^{i-1}{\delta(\r_{k-i}^{\mbox{\scriptsize$\vecalpha$}})\choose
2} + (n-1)^k{\delta(\p)\choose 2} +\sum_{i=1}^{k-1}
(n-1)^{k-i}n^{i-1}{\delta(\p_{k-i}^{\mbox{\scriptsize
$\vecalpha$}})\choose 2}.
$$
\end{pr}
 Now the transitivity coefficient follows from the former two
results and, as Fig. \ref{transfig} shows, tends quickly to zero as
$k\rightarrow \infty$.

\begin{figure}[htbp]
\begin{center}
\includegraphics[width=12cm]{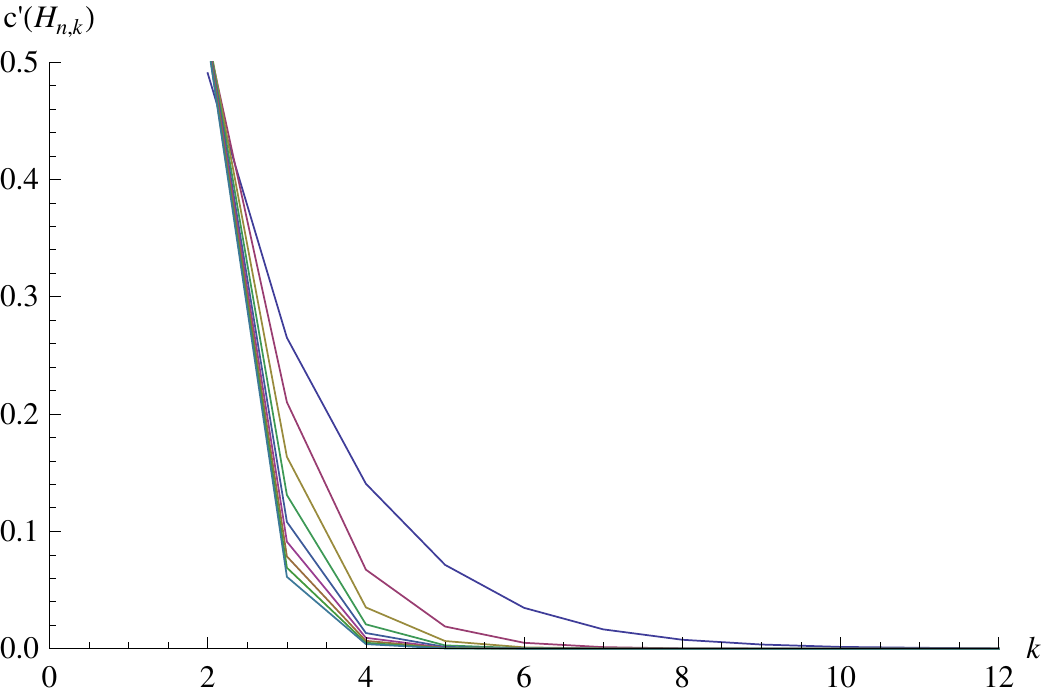}
\caption{Transitivity coefficient of $H_{n,k}$ for
$n=4, 6,\ldots,20$.} \label{transfig}
\end{center}
\end{figure}


\section{Conclusions}

In this paper we have provided a family of graphs that generalize
the hierarchical network introduced in~\cite{RaSoMoOlBa02}, and
combine a modular structure with a scale-free topology, in order to
model modular structures associated to living organisms, social
organizations and technical systems. For the proposed graphs, we have calculated their radius, diameter, degree distribution and clustering coefficient. Moreover,  we have seen that they are scale-free with a power law
exponent, which depends on the initial complete graph; that the
clustering distribution $c(z)$ scales with the degree as $z^{-1}$;
and that the clustering coefficient does not depend on the order of the graph,
as in many networks associated to real
systems~\cite{BaOl04,RaBa03,RaSoMoOlBa02}. Finally, it is worth
mentioning that our definition can be generalized by taking the
vertex set $\Z_{n_1}\times \Z_{n_2}\times \cdots \times \Z_{n_p}$
(instead of $\Z_{n}^p$), so obtaining similar results.

\subsection*{Acknowledgment}
This research was supported by the
{\em Ministerio de Econom\'{\i}a y Competitividad} (Spain) and the {\em European Regional
Development Fund} under project MTM2011-28800-C02-01, and the {\em Catalan Research
Council} under project 2014SGR1147.

\bibliographystyle{plain}

\end{document}